\documentclass[a4paper, 10 pt, conference]{ieeeconf}
\IEEEoverridecommandlockouts
\usepackage{amsfonts,amsmath,amssymb} 

\usepackage{psfrag,color}
\usepackage{enumerate,cite,latexsym,graphicx}
\newtheorem{theorem}{Theorem}
\newtheorem{lemma}{Lemma}
\newtheorem{corollary}{Corollary}
\newtheorem{definition}{Definition}

\newtheorem{remark}{Remark}

\title{\LARGE \bf Quantum Linear Systems Theory
}

\author{Ian R.~Petersen%
\thanks{Research supported by the
Australian Research Council (ARC). A preliminary version of this paper appeared in the 2010 MTNS Conference.
}%
\thanks{Ian R. Petersen is with the School of Engineering and Information Technology, 
        University of New South Wales at the Australian Defence Force Academy, Canberra ACT 2600, Australia.
         {\tt\small i.r.petersen@gmail.com} } }%

\begin{document}
\maketitle
\thispagestyle{empty}
\pagestyle{empty}

\begin{abstract}
This paper surveys  some recent results on the
theory of quantum linear systems and presents them within a unified framework. Quantum linear systems are a class
of systems whose dynamics, which are described by the laws of quantum
mechanics, take the specific form of a set of linear quantum stochastic
differential equations (QSDEs). Such systems commonly arise in the
area of quantum optics and related disciplines. Systems whose dynamics
can be described or approximated by linear QSDEs include
interconnections of optical cavities, beam-spitters, phase-shifters,
optical parametric amplifiers, optical squeezers, and cavity quantum
electrodynamic systems.  With advances in quantum technology, the
feedback control of such quantum systems is generating new challenges
in the field of control theory. Potential applications of such quantum
feedback control systems include quantum computing, quantum error
correction, quantum communications, gravity wave detection, metrology,
atom lasers, and superconducting quantum circuits.

A recently emerging approach to the feedback control of quantum linear
systems involves the use of a controller which itself is a quantum linear
system. This approach to quantum feedback control, referred to as
coherent quantum feedback control, has the advantage that it does not
destroy quantum information, is fast, and has the potential for
efficient implementation. This paper discusses recent
results concerning the synthesis of H-infinity optimal
controllers for linear quantum systems  in the coherent control
case. An important issue which 
arises both in the modelling of linear quantum systems and in the
synthesis of linear coherent quantum controllers is the issue of
physical realizability. This issue relates to the property of whether
a given set of QSDEs corresponds to a physical quantum system
satisfying the laws of quantum mechanics. The paper will cover
recent results relating the question of physical realizability to
notions occuring in linear systems theory such as lossless bounded
real systems and dual J-J unitary systems. 
\end{abstract}

\section{Introduction} \label{sec:intro}

Developments in quantum  technology and quantum information 
provide an important motivation for research in the area of quantum
feedback control systems; e.g., see
\cite{BEL83,WIS94,DJ99,AASDM02,YK03a,YK03b,DP3}. In particular, 
in recent years, there has been considerable interest in the feedback
control and modeling of
linear quantum systems; e.g., see
\cite{DJ99,YK03a,YK03a,EB05,YAM06,BE08,JNP1,NJP1,GGY08,MaP1a,MaP3,MaP4,MAB08,YNJP1,PET08Aa,MaP2a,ShP5a,NJD09,GJ09,RD09,GJN10,WM10}.
Such linear quantum  
systems commonly arise in the area of quantum optics; e.g., see
\cite{WM94,GZ00,BR04}. Feedback control of quantum optical systems has
applications in areas such as quantum communications, quantum 
teleportation, and gravity wave detection. In
particular, \emph{linear quantum optics} is one of the possible
platforms being investigated for future communication systems
(see \cite{JPF02,KWD03}) and quantum computers (see \cite{KLM01,KBSM09} and
\cite{NC00}). Feedback control of quantum systems aims to achieve
closed loop properties such as stability \cite{SoP1a,SoP2a},
robustness \cite{JNP1,SPJ1a}, entanglement \cite{YNJP1,HP4a,NPJ1}. 

Quantum linear system models have been used in the physics
and mathematical physics literature since the 1980's; e.g., see \cite{HP84,GC85a,PAR92,GZ00,WM10}. An important class of linear quantum stochastic models describe the Heisenberg evolution of the (canonical) position and momentum, or annihilation and creation operators of several independent open quantum harmonic oscillators that are coupled to external coherent bosonic fields, such as coherent laser beams; e.g.,  see \cite{WM94}, \cite{WM10}, \cite{GZ00}, \cite{EB05}, \cite{BE08}, \cite{WD05,JNP1,NJP1,NJD09,YAM06,MAB08,YNJP1,GGY08,SSM08,GJN10}).
These linear stochastic models describe quantum optical devices such as optical cavities \cite{BR04}, \cite{WM94}, linear quantum amplifiers \cite{GZ00}, and finite bandwidth squeezers \cite{GZ00}. Following  \cite{JNP1,NJP1,NJD09}, we will refer to this class of models as {\em linear quantum stochastic systems}. In particular, we consider  linear quantum stochastic differential equations
driven by quantum Wiener processes; see \cite{GZ00}. Further details on
quantum stochastic differential equations and quantum  Wiener
processes can be found in  \cite{HP84,PAR92,BHJ07}.

This paper will survey some of the available results on the feedback
control of linear quantum systems and related problems. An important class of quantum feedback control systems involves the use of measurement devices to obtain classical output signals from the quantum system and no quantum measurements is involved. These classical signals are fed into a classical controller which may be implemented via analog or digital electronics and then the resulting control signal act on the quantum system via an actuator. However,  some recent papers on the feedback control of linear quantum systems 
have considered the case in which the feedback controller  itself is
also a
quantum system. Such feedback control is often
referred to as coherent quantum control; e.g., see
\cite{WM94a,SL00,YK03a,YK03b,JNP1,NJP1,MaP1a,MaP3,MaP4,MAB08,GW09}. Due to the
limitations 
imposed by quantum mechanics on the use of quantum measurement, the use of
coherent quantum feedback control may lead to improved control system
performance. In addition, in many applications, coherent quantum
feedback controllers may be preferable to classical feedback controllers
 due to considerations of speed and ease of implementation.  

One motivation for considering such coherent
quantum control problems is that coherent controllers have the
potential to achieve improved performance since quantum measurements
inherently involve the destruction of quantum information; e.g., see
\cite{NC00}. Also, technology is emerging which will enable the
implementation of complex coherent quantum controllers (e.g., see
\cite{PCRYO06}) and the coherent $H^\infty$ controllers proposed in
\cite{JNP1} have already been implemented experimentally as described
in \cite{MAB08}. Furthermore, coherent controllers implemented using
quantum optics have the potential to operate at much higher speeds
than classical controllers implemented in analog or digital
electronics.

In general, quantum linear stochastic systems represented by
linear Quantum Stochastic Differential Equations (QSDEs) with
arbitrary constant coefficients need not correspond to physically
meaningful systems. In contrast, because classical linear
stochastic systems can be implemented at least approximately,
using analog or digital electronics, we regard them as always
being realizable. Physical quantum systems must satisfy some
additional constraints that restrict the allowable values for the
system matrices defining the QSDEs. In particular, the laws of
quantum mechanics dictate that closed quantum systems evolve
\emph{unitarily}, implying that (in the Heisenberg picture)
certain canonical observables satisfy the so-called canonical
commutation relations (CCR) at all times. Therefore, to
characterize physically meaningful systems,  \cite{JNP1} has
introduced a formal notion of physically realizable quantum linear
stochastic systems and derives a pair of necessary and sufficient
characterizations for such systems in terms of constraints on
their system matrices.

In the paper \cite{ShP5a},  the physical realizability results of
\cite{MaP1a,MaP3} are extended to the most general class of complex linear
QSDEs. It is  shown that this
class of linear quantum systems corresponds to the class of real linear
quantum systems considered in \cite{JNP1} via the use of a suitable state
transformation. 

The remainder of this paper proceeds as follows. In Section
\ref{sec:systems}, we introduce the class of linear quantum stochastic
systems under consideration and consider a number of different
representations of these systems. We also introduce a useful special class of
linear quantum systems which was considered in \cite{MaP1a,MaP3,MaP4}. In
Section \ref{sec:physical_realizability}, we consider the issue of
physical realizability for the class of linear quantum systems under
consideration. In Section \ref{sec:Hinf}, we will consider
the problem of coherent $H^\infty$ quantum controller synthesis. In
Section \ref{sec:conclusions}, we present some conclusions.

\section{Linear Quantum System Models}
\label{sec:systems}
In this section, we formulate the class of linear quantum system
models under consideration. These linear quantum system models take
the form of quantum stochastic differential equations which are
derived from the quantum harmonic oscillator. 
\subsection{Quantum Harmonic Oscillators}
\label{subsec:harmonic_oscillator}
We begin by considering a collection of $n$ independent quantum
harmonic oscillators which are
defined on a Hilbert space $\mathcal{H} =
L^2(\mathbb{R}^n,\mathbb{C})$; e.g., see
\cite{MEY95,PAR92,GJN10}. Elements of the Hilbert space $\mathcal{H}$,
$\psi(x)$ are the standard complex valued wave functions arising in
quantum mechanics where $x$ is a spatial variable.  Corresponding to
this collection of  
harmonic oscillators is a vector of  {\em annihilation operators}
\begin{equation}
\label{anninil}
a = \left[\begin{array}{c}a_1\\a_2\\ \vdots \\ a_n \end{array}\right].
\end{equation}
Each annihilation operator $a_i$  is an
unbounded linear operator defined on a suitable domain in $\mathcal{H}$ by 
\[
(a_i\psi)(x) = \frac{1}{\sqrt{2}}x_i\psi(x) +
\frac{1}{\sqrt{2}}\frac{\partial \psi(x)}{\partial x_i}
\]
where $\psi \in \mathcal{H}$ is contained in the  domain of the
operator $a_i$. The adjoint of the operator $a_i$ is denoted $a_i^*$
and is referred to as a creation operator. The operators $a_i$ and
$a_i^*$ are such that the following cannonical commutation relations are satisfied 
\begin{equation}
\label{CCR1}
[a_i,a_j^*] = a_ia_j^*-a_j^*a_i = \delta_{ij}
\end{equation}
where $\delta_{ij}$ denotes the kronecker delta multiplied by the
identity operator on the Hilbert space $\mathcal{H}$. We also have the
commutation relations
\begin{equation}
\label{CCR1A}
[a_i,a_j] = 0, ~[a_i^*,a_j^*]=0.
\end{equation}
For a general vector of operators
\[
g =  
 \left[\begin{array}{c}g_1\\g_2\\ \vdots \\ g_n \end{array}\right],
\]
on $\mathcal{H}$, we use the notation 
\[
g^\# =  
 \left[\begin{array}{c}g^*_1\\g^*_2\\ \vdots \\ g^*_n \end{array}\right],
\]
to denote the corresponding vector of adjoint operators. Also, $g^T$
denotes the corresponding row vector of operators $g^T =
\left[\begin{array}{cccc}g_1 & g_2 & \ldots &g_n \end{array}\right]$,
and $g^\dagger = \left(g^\#\right)^T$. Using this notation, the
canonical commutation relations (\ref{CCR1}), (\ref{CCR1A}) can be written as 
\begin{eqnarray}
\label{CCR2}
\left[\left[\begin{array}{l}
      a\\a^\#\end{array}\right],\left[\begin{array}{l}
      a\\a^\#\end{array}\right]^\dagger\right]
&=&\left[\begin{array}{l} a\\a^\#\end{array}\right]
\left[\begin{array}{l} a\\a^\#\end{array}\right]^\dagger
\nonumber \\
&&- \left(\left[\begin{array}{l} a\\a^\#\end{array}\right]^\#
\left[\begin{array}{l} a\\a^\#\end{array}\right]^T\right)^T\nonumber \\
&=& \left[\begin{array}{cc}I & 0\\
0 & -I\end{array}\right] .
\end{eqnarray}

A state on our system of quantum harmonic oscillators is defined by a
density operator $\rho$ which is a self-adjoint positive-semidefinite
operator on $\mathcal{H}$ with $tr(\rho) =1$; e.g., see \cite{NC00}.
Corresponding to a state $\rho$ and an operator $g$ on $\mathcal{H}$ is the quantum expectation
\[
\langle g \rangle = tr(\rho g).
\]

A state on the system is said to be {\em Gaussian} with
positive-semidefinite {\em covariance
matrix} $Q \in \mathbb{C}^{2n\times 2n}$ and {\em mean vector} $\alpha \in
\mathbb{C}^{n}$ if given any vector $u \in \mathbb{C}^n$, 
\begin{eqnarray*}
\lefteqn{\left\langle \exp\left(i
\left[\begin{array}{cc}u^\dagger & u^T\end{array}\right]
\left[\begin{array}{c}a \\ a^\#\end{array}\right]
\right)\right\rangle} \\
&=& \exp\left(\begin{array}{c}-\frac{1}{2}\left[\begin{array}{cc}u^\dagger &
      u^T\end{array}\right]Q
\left[\begin{array}{c}u \\ u^\#\end{array}\right] \\
+ i \left[\begin{array}{cc}u^\dagger & u^T\end{array}\right] 
\left[\begin{array}{c}\alpha  \\ \alpha^\#\end{array}\right];
\end{array}\right);
\end{eqnarray*}
e.g., see \cite{GJN10,MEY95}. Here, $u^\#$ denotes the complex conjugate of the complex vector
$u$, $u^T$ denotes the transpose of the complex vector $u$, and
$u^\dagger$ denotes the complex conjugate transpose of the complex
vector $u$. 

Note that in the zero mean case, $\alpha = 0$, the covariance matrix $Q$ satisfies
\[
Q = \left\langle \left[\begin{array}{c}a \\ a^\#\end{array}\right] \left[\begin{array}{c}a \\ a^\#\end{array}\right]^\dagger\right\rangle.
\]
In the special case in which the covariance matrix $Q$ is of the form
\[
Q = \left[\begin{array}{cc}I & 0 \\
0 & 0 \end{array}\right]
\]
and the mean $\alpha =0$, the system is said to be in the {\em vacuum
  state}. In the sequel, it will be assumed that the state on the
system of harmonic oscillators is a Gaussian vacuum state. The state on the
system of harmonic oscillators plays a similar role to the probability
distribution of the initial conditions of a classical stochastic
system. 

The quantum harmonic oscillators described above are assumed to be
coupled to $m$ external independent quantum fields modelled by bosonic
annihilation field operators $\mathcal{A}_1(t),
\mathcal{A}_2(t),\ldots,\mathcal{A}_m(t)$ which are defined on
separate Fock spaces $\mathcal{F}_i$ defined  over
$L^2(\mathbb{R})$ for each field operator
\cite{HP84,PAR92,BHJ07,NPJ1}. 
For each annihilation field operator
$\mathcal{A}_j(t)$, there is a corresponding creation field operator
$\mathcal{A}_j^*(t)$, which is defined on the same Fock space and is
the operator adjoint of $\mathcal{A}_j(t)$. 
The field operators are adapted quantum stochastic processes with
forward differentials
\[
d\mathcal{A}_j(t)=\mathcal{A}_j(t+dt)-\mathcal{A}_j(t)
\] and
\[
d\mathcal{A}_j^*(t)=\mathcal{A}_{j}^*(t+dt)-\mathcal{A}_j^*(t)
\]
 that
have the quantum It\^{o} products \cite{HP84,PAR92,BHJ07,NPJ1}:
\begin{eqnarray*}
d\mathcal{A}_{j}(t)d\mathcal{A}_{k}(t)^*&=&\delta_{jk}dt; \\
d\mathcal{A}_{j}^*(t)d\mathcal{A}_{k}(t)&=& 0; \\
d\mathcal{A}_{j}(t)d\mathcal{A}_{k}(t)&=&0; \\
d\mathcal{A}_{j}^*(t)
d\mathcal{A}_{k}^*(t)
&=&0. 
\end{eqnarray*}
The field annilation operators are also collected into a vector of
operators defined as follows:
\[
 \mathcal{A}(t)=\left[\begin{array}{c}
\mathcal{A}_1(t)\\ \mathcal{A}_2(t)\\ \vdots \\ \mathcal{A}_m(t)
\end{array}\right].
\]

For each $i$, the corresponding system state on the Fock space
$\mathcal{F}_i$ is assumed to be a Gaussian vacuum state which means
that given any complex valued function $u_i(t) \in L^2(\mathbb{R},\mathbb{C})$, then
\begin{eqnarray*}
\lefteqn{\left\langle \exp\left(i
\int_0^\infty u_i(t)^* d\mathcal{A}_i(t)
+ i\int_0^\infty u_i(t)d\mathcal{A}_i(t)^*
\right)\right\rangle} \\
&=& \exp\left(-\frac{1}{2}\int_0^\infty |u(t)|^2 dt \right);\hspace{2cm}
\end{eqnarray*}
e.g., see \cite{HP84,PAR92,BHJ07,GJN10}.

In order to describe the joint evolution of the quantum harmonic
oscillators and quantum fields, we first specify the {\em Hamiltonian
operator} for the quantum system which is a Hermitian operator on
$\mathcal{H}$ of the form
\[
H = \frac{1}{2}\left[\begin{array}{cc}a^\dagger &
      a^T\end{array}\right]M
\left[\begin{array}{c}a \\ a^\#\end{array}\right]
\]
where $M \in \mathbb{C}^{2n\times 2n}$ is a Hermitian matrix of the
form
\[
M= \left[\begin{array}{cc}M_1 & M_2\\
M_2^\# &     M_1^\#\end{array}\right]
\]
and $M_1 = M_1^\dagger$, $M_2 = M_2^T$. Here,
$M^\dagger$ denotes the complex conjugate transpose of the
complex matrix $M$, $M^T$ denotes the  transpose of the
complex matrix $M$, and $M^\#$ denotes the complex conjugate  of the
complex matrix $M$. Also, we specify the {\em coupling operator} for
the quantum system to be an operator of the form
\[
L = \left[\begin{array}{cc}N_1 & N_2 \end{array}\right]
\left[\begin{array}{c}a \\ a^\#\end{array}\right]
\]
where $N_1 \in \mathbb{C}^{m\times n}$ and $N_2 \in
\mathbb{C}^{m\times n}$. Also, we write
\[
\left[\begin{array}{c}L \\ L^\#\end{array}\right] = N
\left[\begin{array}{c}a \\ a^\#\end{array}\right] =
\left[\begin{array}{cc}N_1 & N_2\\
N_2^\# &     N_1^\#\end{array}\right]
\left[\begin{array}{c}a \\ a^\#\end{array}\right].
\]
In addition, we define a {\em scattering
  matrix} which is a unitary matrix $S \in \mathbb{C}^{n\times
  n}$. These quantities then define the joint evolution of the quantum
harmonic oscillators and the quantum fields according to a unitary
adapted process $U(t)$ (which is an operator valued function of time) satisfying the Hudson-Parthasarathy QSDE \cite{HP84,PAR92,BHJ07,GJ09}:
\begin{eqnarray*}
dU(t) &=& \left(\mbox{tr}\left[(S-I)^T d\Lambda(t)\right] +  d\mathcal{A}(t)^{\dagger} L -
L^{\dagger}Sd\mathcal{A}(t)\right. \\
&&-\left.(i H + \frac{1}{2}L^{\dagger}L )dt\right)U(t);~U(0) = I,
\end{eqnarray*}
where
$\Lambda(t)=[\Lambda_{jk}(t)]_{j,k=1,\ldots,m}$. Here, the processes
$\Lambda_{jk}(t)$ for $j,k=1,\ldots,m$ are adapted quantum stochastic
processes referred to as {\em gauge processes}, and the forward
differentials $d\Lambda_{jk}(t)=\Lambda_{jk}(t+dt)-\Lambda_{jk}(t) $
$j,k=1,\ldots,m$ have the quantum It\^{o} products: 
\begin{eqnarray*}
d\Lambda_{jk}(t)
d\Lambda_{j'k'}(t)&=&\delta_{kj'}d\Lambda_{jk'}(t); \\
 d\mathcal{A}_j(t)d\Lambda_{kl}(t)&=&\delta_{jk}d\mathcal{A}_l(t);\\
d\Lambda_{jk} d\mathcal{A}_l(t)^*&=&\delta_{kl}d\mathcal{A}_j^*(t).
\end{eqnarray*}
Then, using the Heisenberg picture of quantum mechanics, the harmonic oscillator
operators $a_i(t)$ evolve with time unitarily according to 
$$a_i(t) =
U(t)^*a_iU(t)$$
 for $i=1,2,\ldots,n$. Also, the linear quantum system
output fields are given by 
$$\mathcal{A}^{out}_i(t) =
U(t)^*\mathcal{A}_i(t)U(t)$$
 for $i=1,2,\ldots,m$. 

We now use the fact that for any adapted processes $X(t)$ and $Y(t)$ satisfying a quantum It\^{o}
stochastic differential equation, we have the {\em quantum It\^{o} rule}
$$dX(t)Y(t)=X(t)dY(t)+dX(t) Y(t) + dX(t) dY(t);$$
 e.g., see
\cite{PAR92}.  Using the quantum It\^{o} rule and the quantum It\^{o} products
given above, as well as exploiting the canonical commutation relations
between the operators in $a$, the following QSDEs decribing the linear
quantum system can be obtained (e.g., see \cite{GJN10}):
\begin{eqnarray}
\label{qsde1}
da(t)&=& dU(t)^* a U(t)\nonumber \\
 &=& \left[\begin{array}{cc}F_1 &
    F_2\end{array}\right]
\left[\begin{array}{l} a(t)\\a(t)^\#\end{array}\right]dt \nonumber \\
&&+ \left[\begin{array}{cc}G_1 &  G_2\end{array}\right]
 \left[\begin{array}{l} d\mathcal{A}(t)
\\ d\mathcal{A}(t)^{\#} \end{array}\right];  \nonumber \\
a(0)&=&a; \nonumber \\
d\mathcal{A}^{out}(t)&=& dU(t)^* \mathcal{A}(t)U(t)\nonumber \\
&=& 
\left[\begin{array}{cc} H_1 & H_2 \end{array}\right]
\left[\begin{array}{l} a(t)\\a(t)^\#\end{array}\right]dt\nonumber \\
&&+  \left[\begin{array}{cc}K_1 & K_2 \end{array}\right]
 \left[\begin{array}{l} d\mathcal{A}(t)
\\ d\mathcal{A}(t)^{\#} \end{array}\right],
\end{eqnarray}
where
\begin{eqnarray}
\label{canonicalFGHK}
F_1 &=& -i M_1 -\frac{1}{2}
\left(N_1^\dagger N_1 - N_2^T N_2^\#\right);\nonumber  \\
F_2 &=& -i M_2 -\frac{1}{2}
\left(N_1^\dagger N_2 - N_2^T N_1^\#\right);\nonumber \\ 
G_1 &=& -N_1^\dagger S;\nonumber \\
G_2 &=& N_2^TS^\#;\nonumber \\
H_1 &=& N_1;\nonumber \\
H_2 &=& N_2;\nonumber \\
K_1 &=& S;\nonumber \\
K_2 &=& 0.
\end{eqnarray}
From this, we can write
\begin{eqnarray}
\label{qsde1a}
\left[\begin{array}{l} da(t)\\da(t)^\#\end{array}\right] &=& 
F\left[\begin{array}{l} a(t)\\a(t)^\#\end{array}\right]dt 
+ G \left[\begin{array}{l} d\mathcal{A}(t)
\\ d\mathcal{A}(t)^{\#} \end{array}\right];  \nonumber \\
\left[\begin{array}{l} d\mathcal{A}^{out}(t)
\\ d\mathcal{A}^{out}(t)^{\#} \end{array}\right] &=& 
H \left[\begin{array}{l} a(t)\\a(t)^\#\end{array}\right]dt 
+ K \left[\begin{array}{l} d\mathcal{A}(t)
\\ d\mathcal{A}(t)^{\#} \end{array}\right],\nonumber \\
\end{eqnarray}
where
\begin{eqnarray}
\label{FGHKform}
F &=& \left[\begin{array}{cc}F_1  & F_2\\
F_2^\# & F_1^\#\end{array}\right]; ~~
G = \left[\begin{array}{cc}G_1  & G_2\\
G_2^\# & G_1^\#\end{array}\right]; \nonumber \\
H &=& \left[\begin{array}{cc}H_1  & H_2\\
H_2^\# & H_1^\#\end{array}\right]; ~~
K = \left[\begin{array}{cc}K_1  & K_2\\
K_2^\# & K_1^\#\end{array}\right]. ~~
\end{eqnarray}
Also, the equations (\ref{canonicalFGHK}) can be re-written as
\begin{eqnarray}
\label{canonicalFGHK1}
 F &=& -i J M -\frac{1}{2}  J N^\dagger J N; \nonumber \\
 G &=& -  J N^\dagger 
\left[\begin{array}{cc}S & 0 \\
0 & -S^\#\end{array}\right]; \nonumber \\
 H &=&  N; \nonumber \\
 K &=& \left[\begin{array}{cc}S & 0 \\
0 & S^\#\end{array}\right];
\end{eqnarray}
where 
\[
J = \left[\begin{array}{cc}I & 0 \\
0 & -I\end{array}\right].
\]

Note that matrices of the form (\ref{FGHKform}) occur commonly in the
theory of linear quantum systems. It is straightforward to establish
the following lemma which characterizes matrices of this form.

\begin{lemma}
\label{L1}
A matrix $R = \left[\begin{array}{cc}R_1  & R_2\\
R_3 & R_4\end{array}\right]$ satisfies 
\[
\left[\begin{array}{cc}R_1  & R_2\\
R_3 & R_4\end{array}\right] = \left[\begin{array}{cc}R_1  & R_2\\
R_2^\# & R_1^\#\end{array}\right]
\]
if and only if
\[
R\Sigma = \Sigma R^\#
\]
where
\[
\Sigma = \left[\begin{array}{cc}0 & I \\
I & 0\end{array}\right].
\] 
\end{lemma}

We now consider the case when the initial condition in the QSDE
(\ref{qsde1}) is no longer the vector of annihilation operators
(\ref{anninil}) but rather a vector of linear combinations of
annihilation operators and creation operators defined by
\[
\tilde a = T_1 a + T_2 a^\#
\]
where 
\[
T = \left[\begin{array}{cc} T_1 & T_2\\
T_2^\# & T_1^\#\end{array}\right]\in \mathbb{C}^{2n\times 2n}
\] is non-singular. Then, it follows
from (\ref{CCR2}) that
\begin{eqnarray*}
\lefteqn{\left[\left[\begin{array}{l}
      \tilde a\\\tilde a^\#\end{array}\right],\left[\begin{array}{l}
      \tilde a\\\tilde a^\#\end{array}\right]^\dagger\right]}\nonumber \\
 &=& \left[\begin{array}{l} \tilde a\\\tilde a^\#\end{array}\right]
\left[\begin{array}{l} \tilde a\\\tilde a^\#\end{array}\right]^\dagger
- \left(\left[\begin{array}{l} \tilde a\\\tilde a^\#\end{array}\right]^\#
\left[\begin{array}{l} \tilde a\\\tilde a^\#\end{array}\right]^T\right)^T
\nonumber \\
&=& T\left[\begin{array}{l} \tilde a\\\tilde a^\#\end{array}\right]
\left[\begin{array}{l} \tilde a\\\tilde a^\#\end{array}\right]^\dagger
T^\dagger \nonumber \\
&&- \left(T^\# \left(\left[\begin{array}{l}  a\\ a^\#\end{array}\right]^\#
\left[\begin{array}{l}  a\\ a^\#\end{array}\right]^T\right)^TT^T\right)^T\\
&=& T\left(\begin{array}{c}\left[\begin{array}{l} a\\a^\#\end{array}\right]
\left[\begin{array}{l} a\\a^\#\end{array}\right]^\dagger \\
- \left(\left[\begin{array}{l} a\\a^\#\end{array}\right]^\#
\left[\begin{array}{l} a\\a^\#\end{array}\right]^T\right)^T\end{array}\right)T^\dagger \\
&=& \Theta
\end{eqnarray*}
where 
\begin{eqnarray}
\label{Theta}
\Theta &=&  TJT^\dagger \nonumber \\
&=&\left[\begin{array}{cc}T_1T_1^\dagger - T_2T_2^\dagger & T_1T_2^T -T_2T_1^T\\
T_2^\#T_1^\dagger -T_1^\# T_2^\dagger& T_2^\#T_2^T-T_1^\#T_1^T\end{array}\right].
\end{eqnarray}
The relationship 
\begin{equation}
\label{CCR3}
\left[\left[\begin{array}{l}
      \tilde a\\\tilde a^\#\end{array}\right],\left[\begin{array}{l}
      \tilde a\\\tilde a^\#\end{array}\right]^\dagger\right] = \Theta
\end{equation}
is referred to as a {\em generalized commutation relation}
\cite{MaP1a,MaP3,MaP4}. Also, the covariance matrix corresponding to
$\left[\begin{array}{l}
      \tilde a\\\tilde a^\#\end{array}\right]$ is given by
\begin{eqnarray*}
\tilde Q  &=& \left\langle \left[\begin{array}{c}\tilde a \\
      \tilde a^\#\end{array}\right] \left[\begin{array}{c}\tilde a \\
      \tilde a^\#\end{array}\right]^\dagger\right\rangle \\
 &=& T\left[\begin{array}{cc}I & 0\\
0 & 0\end{array}\right]T^\dagger \nonumber \\
&=&\left[\begin{array}{cc}T_1T_1^\dagger & T_1T_2^T \\
T_2^\#T_1^\dagger & T_2^\#T_2^T\end{array}\right].
\end{eqnarray*}

In terms of the variables $\tilde a(t) =
U(t)^*\tilde aU(t)$, the QSDEs, (\ref{qsde1a}) can be rewritten as
\begin{eqnarray}
\label{qsde3}
\left[\begin{array}{l} d\tilde a(t)\\d\tilde a(t)^\#\end{array}\right] &=& 
\tilde F\left[\begin{array}{l} \tilde a(t)\\\tilde a(t)^\#\end{array}\right]dt 
+ \tilde G \left[\begin{array}{l} d\mathcal{A}(t)
\\ d\mathcal{A}(t)^{\#} \end{array}\right];  \nonumber \\
\left[\begin{array}{l} d\mathcal{A}^{out}(t)
\\ d\mathcal{A}^{out}(t)^{\#} \end{array}\right] &=& 
\tilde H \left[\begin{array}{l} \tilde a(t)\\\tilde a(t)^\#\end{array}\right]dt 
+ \tilde K \left[\begin{array}{l} d\mathcal{A}(t)
\\ d\mathcal{A}(t)^{\#} \end{array}\right],\nonumber \\
\end{eqnarray}
where
\begin{eqnarray}
\label{tildeFGHK1}
\tilde F &=& \left[\begin{array}{cc}\tilde F_1  & \tilde F_2\\
\tilde F_2^\# & \tilde F_1^\#\end{array}\right] = TFT^{-1}; \nonumber \\
\tilde G &=& \left[\begin{array}{cc}\tilde G_1  & \tilde G_2\\
\tilde G_2^\# & \tilde G_1^\#\end{array}\right] = TG; \nonumber\\
\tilde H &=& \left[\begin{array}{cc}\tilde H_1  & \tilde H_2\\
\tilde H_2^\# & \tilde H_1^\#\end{array}\right]=HT^{-1}; \nonumber \\
\tilde K &=& \left[\begin{array}{cc}\tilde K_1  & \tilde K_2\\
\tilde K_2^\# & \tilde K_1^\#\end{array}\right] = K. 
\end{eqnarray}

Now, we can re-write the operators $H$ and $L$ defining the above
collection of quantum harmonic oscillators in terms of the variables
$\tilde a$ as 
\[
H = \frac{1}{2}\left[\begin{array}{cc}\tilde a^\dagger &
      \tilde a^T\end{array}\right]\tilde M
\left[\begin{array}{c}\tilde a \\ \tilde a^\#\end{array}\right], ~~
\left[\begin{array}{c}\tilde L \\ \tilde L^\#\end{array}\right] = \tilde N
\left[\begin{array}{c}\tilde a \\ \tilde a^\#\end{array}\right]
\]
where 
\begin{equation}
\label{tildeHL}
\tilde M = \left(T^\dagger\right)^{-1} M T^{-1}, ~~
\tilde N = NT^{-1}.
\end{equation}
Here,
\begin{equation}
\label{tildeMN}
\tilde M= \left[\begin{array}{cc}\tilde M_1 & \tilde M_2\\
\tilde M_2^\# &     \tilde M_1^\#\end{array}\right],~~\tilde N = \left[\begin{array}{cc}\tilde N_1 & \tilde
    N_2 \\
\tilde N_2^\# & \tilde N_1^\# \end{array}\right].
\end{equation}
Furthermore, equations (\ref{canonicalFGHK1}), (\ref{tildeFGHK1}) and  (\ref{tildeHL}) can be combined to 
obtain
\begin{eqnarray}
\label{generalizedFGHK1}
\tilde F &=& -i \Theta  \tilde M -\frac{1}{2} \Theta \tilde N^\dagger J \tilde N; \nonumber \\
\tilde G &=& -\Theta  \tilde N^\dagger 
\left[\begin{array}{cc}S & 0 \\
0 & -S^\#\end{array}\right]; \nonumber \\
\tilde H &=& \tilde N; \nonumber \\
\tilde K &=& \left[\begin{array}{cc}S & 0 \\
0 & S^\#\end{array}\right].
\end{eqnarray}
Note that since $S$ is unitary, it follows that
\begin{eqnarray}
\label{KJKdagger}
\tilde  K J \tilde K^\dagger &=& \left[\begin{array}{cc}S & 0 \\
0 & S^\#\end{array}\right] \left[\begin{array}{cc}I & 0 \\
0 & -I\end{array}\right] \left[\begin{array}{cc}S^\dagger & 0 \\
0 & S^T\end{array}\right] \nonumber \\
&=&  \left[\begin{array}{cc}S & 0 \\
0 & S^\#\end{array}\right] \left[\begin{array}{cc}S^\dagger & 0 \\
0 & -S^T\end{array}\right] \nonumber \\
&=& \left[\begin{array}{cc}SS^\dagger & 0 \\
0 & -S^\#S^T\end{array}\right] = J. 
\end{eqnarray}
Also,
\begin{eqnarray}
\label{Kunitary}
\tilde  K \tilde K^\dagger &=& \left[\begin{array}{cc}S & 0 \\
0 & S^\#\end{array}\right]  \left[\begin{array}{cc}S^\dagger & 0 \\
0 & S^T\end{array}\right] \nonumber \\
&=& \left[\begin{array}{cc}SS^\dagger & 0 \\
0 & S^\#S^T\end{array}\right] = I. 
\end{eqnarray}
Indeed these two properties characterize all matrices non-singular $\tilde K$  satisfying $\tilde K \Sigma = \Sigma \tilde K^\#$ which are of the form given in (\ref{generalizedFGHK1}). Let the nonsingular matrix $K$ be such that $ K \Sigma = \Sigma  K^\#$, $KJK^\dagger =J$, and $KK^\dagger =I$. It follows from Lemma \ref{L1} that we can write
\[
 K = \left[\begin{array}{cc}K_1 & K_2 \\
K_2^\# & K_1^\#\end{array}\right].
\]
Also, $KK^\dagger =I$ implies
\begin{eqnarray*}
 KK^\dagger &=& \left[\begin{array}{cc}K_1 & K_2 \\
K_2^\# & K_1^\#\end{array}\right]\left[\begin{array}{cc}K_1^\dagger & K_2^T \\
K_2^\dagger & K_1^T\end{array}\right] \nonumber \\
&=& \left[\begin{array}{cc}K_1K_1^\dagger + K_2K_2^\dagger & K_1K_2^T+K_2K_1^T\\
K_2^\#K_1^\dagger+K_1^\#K_2^\dagger & K_2^\#K_2^T+ K_1^\#K_1^T
\end{array}\right]
\nonumber \\
&=& \left[\begin{array}{cc}I & 0 \\
0 & I\end{array}\right] 
\end{eqnarray*}
and $KJK^\dagger =J$ implies
\begin{eqnarray*}
 KJK^\dagger &=& \left[\begin{array}{cc}K_1 & K_2 \\
K_2^\# & K_1^\#\end{array}\right]\left[\begin{array}{cc}I & 0 \\
0 & -I\end{array}\right]\left[\begin{array}{cc}K_1^\dagger & K_2^T \\
K_2^\dagger & K_1^T\end{array}\right] \nonumber \\
&=& \left[\begin{array}{cc}K_1 & K_2 \\
K_2^\# & K_1^\#\end{array}\right]\left[\begin{array}{cc}K_1^\dagger & K_2^T \\
-K_2^\dagger & -K_1^T\end{array}\right] \nonumber \\
&=& \left[\begin{array}{cc}K_1K_1^\dagger - K_2K_2^\dagger & K_1K_2^T-K_2K_1^T\\
K_2^\#K_1^\dagger-K_1^\#K_2^\dagger & K_2^\#K_2^T- K_1^\#K_1^T
\end{array}\right]
\nonumber \\
&=& \left[\begin{array}{cc}I & 0 \\
0 & -I\end{array}\right].
\end{eqnarray*}
The $(1,1)$ block of these two equations imply $K_1K_1^\dagger + K_2K_2^\dagger = I$ and $K_1K_1^\dagger - K_2K_2^\dagger=I$. Hence $K_2K_2^\dagger = 0$. Therefore, $K_2 = 0$ and $K_1K_1^\dagger=I$. From this, it follows that the matrix $K$ must be of the form given in (\ref{generalizedFGHK1}). 

The QSDEs (\ref{qsde3}), (\ref{tildeFGHK1}), (\ref{generalizedFGHK1})
define the general class of linear quantum systems considered in this
paper. Such quantum systems can be used to model a large range of
devices and networks of devices arising in the area of quantum optics
including optical cavities, squeezers, optical parametric amplifiers,
cavity QED systems, 
beam splitters, and  phase shifters; e.g., see
\cite{DJ99,YK03a,YK03b,JNP1,MAB08,PET08Aa,NJD09,RD09,WM10,WM94,GZ00,BR04,GW09}.

\subsection{Annihilation operator linear quantum systems}
\label{subsec:annihilator}

An important special case of the linear quantum systems (\ref{qsde3}),
(\ref{tildeFGHK1}), (\ref{generalizedFGHK1}) corresponds to the case in
which the Hamiltonian operator $H$ and coupling operator $L$ depend only of
the vector of annihilation operators $a$ and not on the vector of
creation operators $a^\#$. This class of linear quantum systems is
considered in \cite{MaP1a,MaP2a,MaP3,MaP4,PET08Aa,PET09Aa,MAB08} and
can be used to model ``passive'' quantum optical devices such as
optical cavities, beam splitters, phase shifters and
interferometers. 

This class of linear quantum systems corresponds to
the case in which $\tilde M_2 = 0$, $\tilde N_2 = 0$, and $T_2 =
0$. In this case, 
the linear quantum system can be modelled by the QSDEs
\begin{eqnarray}
\label{qsde4}
d \tilde a(t) &=& \tilde F \tilde a(t) dt + \tilde G d\mathcal{A}(t)
\nonumber \\
d\mathcal{A}^{out}(t)  &=& \tilde H \tilde a(t) dt + \tilde K d\mathcal{A}(t)
\end{eqnarray}
where
\begin{eqnarray}
\label{annihilFGHK}
\tilde F &=& -i \Theta_1 \tilde M_1 -\frac{1}{2}\Theta_1
\tilde N_1^\dagger \tilde N_1;\nonumber  \\ 
\tilde G &=& -\Theta_1 \tilde N_1^\dagger S;\nonumber \\
\tilde H &=& \tilde N_1;\nonumber \\
\tilde K &=& S; \nonumber \\
\Theta_1 &=& T_1T_1^\dagger > 0.
\end{eqnarray}

\subsection{Position and momentum operator linear quantum systems}
\label{subsec:real}
Note that the matrices in the general QSDEs (\ref{qsde3}),
(\ref{tildeFGHK1}) are in general 
complex. However, it is possible to apply a particular change of
variables to the system (\ref{qsde1}) so that all of the matrices in
the resulting transformed QSDEs are real. This change of variables is
defined as follows:
\begin{eqnarray}
\label{posmom1}
\left[\begin{array}{c}q \\ p \end{array}\right] &=& \Phi
\left[\begin{array}{c}a \\ a^\# \end{array}\right]; \nonumber \\
\left[\begin{array}{c}\mathcal{Q}(t) \\ \mathcal{P}(t) \end{array}\right] &=& \Phi
\left[\begin{array}{c}\mathcal{A}(t) \\
    \mathcal{A}(t)^\# \end{array}\right];\nonumber \\
\left[\begin{array}{c}\mathcal{Q}^{out}(t) \\ \mathcal{P}^{out}(t) \end{array}\right] &=& \Phi
\left[\begin{array}{c}\mathcal{A}^{out}(t) \\ \mathcal{A}^{out}(t)^\# \end{array}\right]
\end{eqnarray}
where the matrices $\Phi$ have the form
\begin{equation}
\label{Phi}
\Phi = \left[\begin{array}{cc}I & I \\
-iI & iI\end{array}\right] 
\end{equation}
and have the appropriate dimensions. 
 Here
$q$ is a vector of the self-adjoint position operators for the system
of harmonic oscillators and $p$ is a vector of momentum operators;
e.g., see \cite{ShP5a,JNP1,NJP1,NPJ1}. Also, $\mathcal{Q}(t)$ and
$\mathcal{P}(t)$ are the vectors of position and momentum operators
for the quantum noise fields acting on the system of harmonic
oscillators. Furthermore, $\mathcal{Q}^{out}(t)$ and
$\mathcal{P}^{out}(t)$ are the vectors of position and momentum operators
for the output quantum noise fields. 

It follows from (\ref{Phi}) that
\begin{equation}
\label{PhiPhidagger}
\Phi \Phi^\dagger = 2I
\end{equation}
and hence 
\begin{equation}
\label{invPhi2}
\Phi^{-1}\Phi^{-\dagger} = \frac{1}{2}I.
\end{equation}

Rather than applying the transformations (\ref{posmom1}) to the
quantum linear system (\ref{qsde1a}) which satisfies the canonical
commutation relations (\ref{CCR2}), corresponding transformations can
be applied to the quantum linear system (\ref{qsde3}) which satisfies
the generalized commutation relations (\ref{CCR3}). These
transformations are as follows:
\begin{eqnarray}
\label{posmom2}
\left[\begin{array}{c}\tilde q \\ \tilde p \end{array}\right] &=& \Phi
\left[\begin{array}{c}\tilde a \\ \tilde a^\# \end{array}\right]; \nonumber \\
\left[\begin{array}{c}\mathcal{Q}(t) \\ \mathcal{P}(t) \end{array}\right] &=& \Phi
\left[\begin{array}{c}\mathcal{A}(t) \\
    \mathcal{A}(t)^\# \end{array}\right];\nonumber \\
\left[\begin{array}{c}\mathcal{Q}^{out}(t) \\ \mathcal{P}^{out}(t) \end{array}\right] &=& \Phi
\left[\begin{array}{c}\mathcal{A}^{out}(t) \\ \mathcal{A}^{out}(t)^\# \end{array}\right].
\end{eqnarray}
When these transformations are applied to the quantum linear system
(\ref{qsde3}), this leads to the following real quantum linear system:
\begin{eqnarray}
\label{qsde5}
\left[\begin{array}{l} d\tilde q(t)\\d\tilde p(t)\end{array}\right] &=& 
A\left[\begin{array}{l} \tilde p(t)\\\tilde q(t)\end{array}\right]dt 
+ B \left[\begin{array}{l} d\mathcal{Q}(t)
\\ d\mathcal{P}(t) \end{array}\right];  \nonumber \\
\left[\begin{array}{l} d\mathcal{Q}^{out}(t)
\\ d\mathcal{P}^{out}(t) \end{array}\right] &=& 
C \left[\begin{array}{l} \tilde q(t)\\\tilde p(t)\end{array}\right]dt 
+ D \left[\begin{array}{l} d\mathcal{P}(t)
\\ d\mathcal{Q}(t) \end{array}\right],\nonumber \\
\end{eqnarray}
where
\begin{eqnarray}
\label{ABCD}
A &=& \Phi \tilde F \Phi^{-1} \nonumber \\
&=&\frac{1}{2}\left[\begin{array}{c}
\tilde F_1 + \tilde F_1^\#+ \tilde F_2 +\tilde F_2^\#  \\
-i\left(\tilde F_1 - \tilde F_1^\#\right) - i \left(\tilde F_2 -\tilde
  F_2^\#\right)  
\end{array}\right. \nonumber \\
&& \hspace{1cm} \left.\begin{array}{c}
i\left(\tilde F_1 - \tilde F_1^\#\right) - i \left(\tilde F_2 -\tilde
  F_2^\#\right) \\
\tilde F_1 + \tilde F_1^\# - \tilde F_2 -\tilde F_2^\#
\end{array}\right]; \nonumber \\
B &=& \Phi \tilde G \Phi^{-1} \nonumber \\
&=&\frac{1}{2}\left[\begin{array}{c}
\tilde G_1 + \tilde G_1^\#+ \tilde G_2 +\tilde G_2^\#  \\
-i\left(\tilde G_1 - \tilde G_1^\#\right) - i \left(\tilde G_2 -\tilde
  G_2^\#\right)  
\end{array}\right. \nonumber \\
&& \hspace{1cm} \left.\begin{array}{c}
i\left(\tilde G_1 - \tilde G_1^\#\right) - i \left(\tilde G_2 -\tilde
  G_2^\#\right) \\
\tilde G_1 + \tilde G_1^\# - \tilde G_2 -\tilde G_2^\#
\end{array}\right]; \nonumber \\
C &=& \Phi \tilde H \Phi^{-1} \nonumber \\
&=&\frac{1}{2}\left[\begin{array}{c}
\tilde H_1 + \tilde H_1^\#+ \tilde H_2 +\tilde H_2^\#  \\
-i\left(\tilde H_1 - \tilde H_1^\#\right) - i \left(\tilde H_2 -\tilde
  H_2^\#\right)  
\end{array}\right. \nonumber \\
&& \hspace{1cm} \left.\begin{array}{c}
i\left(\tilde H_1 - \tilde H_1^\#\right) - i \left(\tilde H_2 -\tilde
  H_2^\#\right) \\
\tilde H_1 + \tilde H_1^\# - \tilde H_2 -\tilde H_2^\#
\end{array}\right]; \nonumber \\
D &=& \Phi \tilde K \Phi^{-1} \nonumber \\
&=&\frac{1}{2}\left[\begin{array}{c}
\tilde K_1 + \tilde K_1^\#+ \tilde K_2 +\tilde K_2^\#  \\
-i\left(\tilde K_1 - \tilde K_1^\#\right) - i \left(\tilde K_2 -\tilde
  K_2^\#\right)  
\end{array}\right. \nonumber \\
&& \hspace{1cm} \left.\begin{array}{c}
i\left(\tilde K_1 - \tilde K_1^\#\right) - i \left(\tilde K_2 -\tilde
  K_2^\#\right) \\
\tilde K_1 + \tilde K_1^\# - \tilde K_2 -\tilde K_2^\#
\end{array}\right]. \nonumber \\
\end{eqnarray}
These matrices are all real. 

Also, it follows
from (\ref{CCR3}) that
\begin{eqnarray}
\label{quad_comm}
\left[\left[\begin{array}{l}
      \tilde q\\\tilde p\end{array}\right],\left[\begin{array}{l}
      \tilde q\\\tilde p\end{array}\right]^\dagger\right]
&=& \Xi
\end{eqnarray}
where 
\begin{equation}
\label{Lambda}
 \Xi = \Phi \Theta \Phi^\dagger = \Phi TJT^\dagger \Phi^\dagger,
\end{equation}
which is a Hermitian matrix. 

Now, we can re-write the operators $H$ and $L$ defining the above
collection of quantum harmonic oscillators in terms of the variables
$\tilde q$ and $\tilde p$ as 
\[
H = \frac{1}{2}\left[\begin{array}{cc}\tilde q^T &
      \tilde p^T\end{array}\right]R
\left[\begin{array}{c}\tilde q \\ \tilde p\end{array}\right], ~~
\left[\begin{array}{c}\tilde L \\ \tilde L^\#\end{array}\right] =  V
\left[\begin{array}{c}\tilde q \\ \tilde p\end{array}\right]
\]
where 
\begin{equation}
\label{realHL}
R = \left(\Phi^\dagger\right)^{-1} \tilde M \Phi^{-1}, ~~
V = \tilde N\Phi^{-1}.
\end{equation}
Here
\begin{eqnarray}
\label{RV}
R
&=&\frac{1}{4}\left[\begin{array}{c}
\tilde M_1 + \tilde M_1^\#+ \tilde M_2 +\tilde M_2^\#  \\
-i\left(\tilde M_1 - \tilde M_1^\#\right) - i \left(\tilde M_2 -\tilde
  \tilde M_2^\#\right)  
\end{array}\right. \nonumber \\
&& \hspace{1cm} \left.\begin{array}{c}
i\left(\tilde M_1 - \tilde M_1^\#\right) - i \left(\tilde M_2 -\tilde
  \tilde M_2^\#\right) \\
\tilde M_1 + \tilde M_1^\# - \tilde M_2 -\tilde M_2^\#
\end{array}\right]; \nonumber \\
V &=& \left[\begin{array}{cc}\tilde N_1 + \tilde N_2 &
i\left(\tilde N_1 -     \tilde N_2\right)\\
\tilde N_1^\# + \tilde N_2^\# &
i\left(\tilde N_2^\# -     \tilde N_1^\#\right)
\end{array}\right]
\end{eqnarray}
where the matrix $R$ is real but the matrix $V$ may be complex. 

However, using (\ref{realHL}) and (\ref{generalizedFGHK1}), we can write
\[
\left[\begin{array}{c}\tilde L+\tilde L^\# \\ \frac{\tilde L -\tilde L^\#}{i}\end{array}\right]
=\Phi\left[\begin{array}{c}\tilde L \\ \tilde L^\#\end{array}\right]  = \Phi V 
\left[\begin{array}{c}\tilde q \\ \tilde p\end{array}\right] = W \left[\begin{array}{c}\tilde q \\ \tilde p\end{array}\right]
\]
where 
\begin{eqnarray*}
W&=& \Phi V = \Phi \tilde N\Phi^{-1}\nonumber \\
&=& \frac{1}{2}\left[\begin{array}{c}
\tilde N_1 + \tilde N_1^\#+ \tilde N_2 +\tilde N_2^\#  \\
-i\left(\tilde N_1 - \tilde N_1^\#\right) - i \left(\tilde N_2 -\tilde
  N_2^\#\right)  
\end{array}\right. \nonumber \\
&& \hspace{1cm} \left.\begin{array}{c}
i\left(\tilde N_1 - \tilde N_1^\#\right) - i \left(\tilde N_2 -\tilde
  N_2^\#\right) \\
\tilde N_1 + \tilde N_1^\# - \tilde N_2 -\tilde N_2^\#
\end{array}\right]
\end{eqnarray*}
is real as in (\ref{ABCD}). That is, we can write
\[
\left[\begin{array}{c}\Re(\tilde L) \\ \Im(\tilde L) \end{array}\right] 
= \frac{1}{2} W \left[\begin{array}{c}\tilde q \\ \tilde p\end{array}\right].
\]

Note that the matrix $\Phi T \Phi^{-1}$ is real and 
\begin{equation}
\label{Jtilde}
\Phi J \Phi^\dagger = 2i\left[\begin{array}{cc}0 & I \\
-I & 0\end{array}\right] = 2i\tilde J
\end{equation}
where
\[
\tilde J = \left[\begin{array}{cc}0 & I \\
-I & 0\end{array}\right].
\]
Hence, the matrix 
\[
\Xi =\Phi T \Phi^{-1} \Phi J \Phi^\dagger
\left(\Phi^\dagger\right)^{-1} T^\dagger \Phi^\dagger
\]
must be purely imaginary. Hence, we can define the real skew symmetric matrix
\[
\tilde \Theta = -\frac{i}{2} \Xi = -\frac{i}{2}\Phi \Theta \Phi^\dagger.
\]
Using this notation, (\ref{quad_comm}) can be written as
\[
\left[\left[\begin{array}{l}
      \tilde q\\\tilde p\end{array}\right],\left[\begin{array}{l}
      \tilde q\\\tilde p\end{array}\right]^\dagger\right] = 2i\tilde \Theta. 
\]

In addition, we note that
\begin{equation}
\label{Jtilde1}
\left(\Phi^\dagger\right)^{-1} J \Phi^{-1} = \frac{i}{2}\tilde J.
\end{equation}

Furthermore, equations (\ref{generalizedFGHK1}), (\ref{ABCD}), and (\ref{RV}) can be combined to 
obtain
\begin{eqnarray}
\label{ABCD1}
A &=& -i \Xi R -\frac{1}{2} \Xi V^\dagger J V \nonumber \\
&=& 2\tilde \Theta R -\frac{1}{2} \Xi W^T \Phi^{-\dagger}J\Phi^{-1}W\nonumber \\
&=& 2\tilde \Theta R-\frac{i}{4}\Xi W^T\tilde JW\nonumber \\
&=& 2\tilde \Theta R+\frac{1}{2}\tilde \Theta W^T\tilde JW;\nonumber \\
B &=& -\Xi  V^\dagger J
\left[\begin{array}{cc}S & 0 \\
0 & S^\#\end{array}\right]\Phi^{-1} \nonumber \\
&=& -\Xi W^T \Phi^{-\dagger}J\Phi^{-1} \Phi \left[\begin{array}{cc}S & 0 \\
0 & S^\#\end{array}\right]\Phi^{-1}\nonumber \\
&=& -\frac{i}{2}\Xi W^T \tilde J D\nonumber \\
&=& \tilde \Theta W^T \tilde J D;\nonumber \\
C &=& W;\nonumber \\
D &=& \frac{1}{2}\left[\begin{array}{cc}S +S^\#& i\left(S-S^\#\right) \\
i\left(S-S^\#\right) & S+S^\#\end{array}\right].
\end{eqnarray}

Now from (\ref{ABCD}), we have $D = \Phi \tilde K \Phi^{-1}$ and $D^T = \Phi^{-\dagger} \tilde K^\dagger \Phi^{\dagger}$ and hence,
\begin{eqnarray}
\label{DDT}
DD^T &=& \Phi \tilde K \Phi^{-1}\Phi^{-\dagger} \tilde K \Phi^{\dagger}\nonumber \\
&=& \frac{1}{2}  \Phi \tilde K  \tilde K^\dagger \Phi^{\dagger}\nonumber \\
&=& \frac{1}{2}  \Phi \Phi^{\dagger}\nonumber \\
&=& I
\end{eqnarray}
using (\ref{invPhi2}), (\ref{Kunitary}) and (\ref{PhiPhidagger}); i.e., $D$ is an orthogonal matrix. 

Also, we have
\begin{eqnarray}
\label{DJDT}
D\tilde J D^T &=& -\frac{i}{2}\Phi \tilde K \Phi^{-1}\Phi J \Phi^\dagger \Phi^{-\dagger} \tilde K \Phi^{\dagger}\nonumber \\
&=& -\frac{i}{2}  \Phi \tilde K J \tilde K^\dagger \Phi^{\dagger}\nonumber \\
&=&- \frac{i}{2}  \Phi J\Phi^{\dagger}\nonumber \\
&=& \tilde J
\end{eqnarray}
using (\ref{Jtilde}) and (\ref{KJKdagger}); i.e., $D$ is a symplectic matrix. 

Conversely suppose a matrix $D= \Phi \tilde K \Phi^{-1}$ satisfies $DD^T = I$ and $D\tilde J D^T = \tilde J$. It follows in a similar fashion to above that $\tilde K \tilde K^\dagger = I$ and $\tilde K J \tilde K^\dagger = J$. Hence, as in Section \ref{subsec:harmonic_oscillator} the matrix $\tilde K$ must be of the form in (\ref{generalizedFGHK1}). Thus, the matrix $D$ must be of the form in (\ref{ABCD1}). 

\section{Physical Realizability}
\label{sec:physical_realizability}
Not all QSDEs of the form (\ref{qsde3}),
(\ref{tildeFGHK1}) correspond to physical quantum systems. This
motivates a notion of physical realizability which has been considered
in the papers
\cite{JNP1,NJP1,MaP1a,MaP2a,MaP3,MaP4,PET08Aa,PET09Aa,ShP5a,HP4a}. 
This notion is of particular importance in the problem of coherent
quantum feedback control in which the controller itself is a quantum
system. In this case, if a controller is synthesized using a method
such as quantum $H^\infty$ control \cite{JNP1,MaP1a,MaP4} or quantum
LQG control \cite{NJP1,HP4a}, it important that the controller can be
implemented as a physical quantum system \cite{NJD09,PET08Aa}. 
We first consider the issue of physical realizability in the case of
general linear quantum systems and then we consider the issue of physical
realizability for the case of annihilator operator linear quantum
system of the form considered in Subsection \ref{subsec:annihilator}.

\subsection{Physical realizability for general linear quantum systems}
\label{subsec:genphysrealiz}
The
formal definition of physically realizable QSDEs requires that they can
be realized as a system of quantum harmonic oscillators. 

\begin{definition}
\label{D1}
 QSDEs of the form (\ref{qsde3}),
(\ref{tildeFGHK1}) are {\em physically realizable} if there exist complex
matrices $\Theta= \Theta^\dagger$, $\tilde M = \tilde M^\dagger$, 
$\tilde N$, $S$ such that $S^\dagger S = I$, $\Theta$ is of the form in
(\ref{Theta}), $\tilde M$ is of the form
in 
(\ref{tildeMN}), and
(\ref{generalizedFGHK1}) is satisfied.  
\end{definition}

A version of the following theorem was presented in \cite{ShP5a}; see also
\cite{NJP1,JNP1} for related results.

\begin{theorem}
\label{T1}
The QSDEs (\ref{qsde3}), (\ref{tildeFGHK1}) are physically realizable if
and only if there exists a complex matrix $\Theta =\Theta^\dagger$
 of the form in
(\ref{Theta}) such that
\begin{eqnarray}
\label{physreal1}
&&\tilde F\Theta 
+ \Theta \tilde F^\dagger + \tilde GJ\tilde G^\dagger = 0;
\nonumber \\
&&\tilde G =  -\Theta \tilde H^\dagger J \tilde K; \nonumber \\
&&\tilde K J \tilde K^\dagger =J; \nonumber \\
&&\tilde K \tilde K^\dagger =I.
\end{eqnarray}
\end{theorem}

\begin{proof}
If there exist matrices $\Theta =\Theta^\dagger$,
$\tilde M = \tilde M^\dagger$, 
$\tilde N$, $S$ such that $S^\dagger S = I$, $\tilde M$ is of the form
in 
(\ref{tildeMN}),  $\Theta$ is of the form in
(\ref{Theta}), and
(\ref{generalizedFGHK1}) is satisfied, then it follows by
straightforward substitution that the first equation in (\ref{physreal1}) will be
satisfied and 
\begin{eqnarray}
\label{physreal1a}
&&\tilde G =  -\Theta \tilde H^\dagger\left[\begin{array}{cc}S & 0 \\
0 & -S^\#\end{array}\right] ; \nonumber \\
&&\tilde K = \left[\begin{array}{cc}S & 0 \\
0 & S^\#\end{array}\right].
\end{eqnarray}
Then, the remaining equations in (\ref{physreal1}) follow using (\ref{KJKdagger}) and (\ref{Kunitary}). 

Converely, suppose  there exists a complex matrix $\Theta =\Theta^\dagger$ of the form in
(\ref{Theta}) such that  (\ref{physreal1}) is
satisfied. Also, as shown at the end of Section II.A, the conditions $\tilde K J \tilde K^\dagger =J$ and $\tilde K \tilde K^\dagger =I$ imply that there exists a complex matrix $S$ such that $S^\dagger S = I$ and $\tilde K$ is of the form $\tilde K = \left[\begin{array}{cc}S & 0 \\
0 & S^\#\end{array}\right].$

Also, let 
\begin{eqnarray*}
\tilde M &=& \frac{i}{2} \left(\Theta^{-1}\tilde F - \tilde F^\dagger \Theta^{-1}\right); \nonumber \\
\tilde N &=& \tilde H.
\end{eqnarray*}
It is straightforward to verify that this matrix $\tilde M$ is
Hermitian. Also, it follows from (\ref{physreal1}) that 
\begin{eqnarray*}
\tilde G
&=&  -\Theta \tilde N^\dagger \left[\begin{array}{cc}S & 0 \\
0 & -S^\#\end{array}\right]
\end{eqnarray*}
 as required. 
Furthermore, using $S^\dagger S = I$, it now follows that 
\[
\tilde G J \tilde G^\dagger = \Theta \tilde N^\dagger J \tilde N \Theta.
\]
Hence, (\ref{physreal1}) implies
\begin{eqnarray*}
\tilde F\Theta+\Theta\tilde F^\dagger
+ \Theta \tilde N^\dagger J \tilde N \Theta = 0
\end{eqnarray*}
and hence
\begin{eqnarray*}
\tilde F^\dagger \Theta^{-1} &=& -\Theta^{-1}\tilde F 
- \tilde N^\dagger J \tilde N 
\end{eqnarray*}
From this, it follows that 
\[
\tilde M = \frac{i}{2} \left(2\Theta^{-1}\tilde F + \tilde N^\dagger J \tilde N \right) 
\]
and hence,
\[
\tilde F = -i \Theta \tilde M -\frac{1}{2} \Theta \tilde N^\dagger J \tilde N
\]
as required. Hence, (\ref{generalizedFGHK1}) is satisfied. 

We now use
Lemma \ref{L1} to show that $\tilde M$ is of the form
in 
(\ref{tildeMN}). Indeed, we have $T\Sigma = \Sigma T^\#$, $T^\#\Sigma
= \Sigma T$, $T^{-1}\Sigma = \Sigma \left(T^\#\right)^{-1}$,
$\left(T^\#\right)^{-1}\Sigma = \Sigma T^{-1} $,  $\tilde F\Sigma =
\Sigma \tilde F^\#$, $\tilde F^\#\Sigma
= \Sigma \tilde F$, and $\Sigma J = -J \Sigma$. Hence, 
\begin{eqnarray*}
\Sigma \tilde M^\# &=& 
-\frac{i}{2}
\left(\begin{array}{c}\Sigma\left(T^T\right)^{-1}J\left(T^\#\right)^{-1}\tilde
    F^\# \\
-\Sigma \tilde F^T
  \left(T^T\right)^{-1}J\left(T^\#\right)^{-1}\end{array}\right)
\nonumber \\
&=& 
\frac{i}{2}
\left(\begin{array}{c}\left(T^\dagger\right)^{-1}JT^{-1}\tilde F \\
-\tilde F^\dagger
  \left(T^\dagger\right)^{-1}JT^{-1}\end{array}\right)\Sigma\nonumber
\\
&=&  \tilde M\Sigma.
\end{eqnarray*}
Therefore, it follows from Lemma \ref{L1} that $\tilde M$ is of the form
in 
(\ref{tildeMN})
 and hence, the QSDEs (\ref{qsde3}),
(\ref{tildeFGHK1}) are physically realizable.
\end{proof}

\begin{remark}
\label{R1}
In the canonical case when $T =I$ and $\Theta = J$, the physical
realizability equations (\ref{physreal1}) become
\begin{eqnarray}
\label{canonphysreal}
&&\tilde FJ 
+ J\tilde F^\dagger + \tilde GJ\tilde G^\dagger = 0;
\nonumber \\
&&\tilde G =  -J \tilde H^\dagger J \tilde K; \nonumber \\
&&\tilde K J \tilde K^\dagger =J; \nonumber \\
&&\tilde K \tilde K^\dagger =I.
\end{eqnarray}
\end{remark}

Following the approach of \cite{ShP5a}, we now relate the physical
realizability of  the QSDEs (\ref{qsde3}), (\ref{tildeFGHK1}) to the
dual $(J,J)$-unitary property of the corresponding transfer function matrix
\begin{equation}
\label{TF}
\Gamma(s) = \left[\begin{array}{cc}\Gamma_{11}(s) & \Gamma_{12}(s) \\
\Gamma_{21}(s) & \Gamma_{22}(s)\end{array}\right]= \tilde H\left(sI-\tilde F\right)^{-1}\tilde G+\tilde K.
\end{equation}

\begin{definition} (See \cite{ShP5a,KIM97}.)
\label{D2}
A transfer function matrix $\Gamma(s)$ of the form (\ref{TF}) is {\em
  dual $(J,J)$-unitary} if 
\[
\Gamma(s)J\Gamma^\sim(s) = J
\]
for all $s \in \mathbb{C}_+$. 
\end{definition}
Here, $\Gamma^\sim(s) =
\Gamma(-s^*)^\dagger$ and $\mathbb{C}_+$ denotes the set $\{s \in
\mathbb{C}: \Re[s] \geq 0\}$. 

\begin{theorem}
\label{T2}
The transfer function matrix (\ref{TF}) corresponding to the QSDEs
(\ref{qsde3}), (\ref{tildeFGHK1}) is 
 dual $ (J,J)$-unitary if and only if
 \begin{equation}
\label{KJKdagger1}
\tilde K J \tilde K^\dagger =   J,
 \end{equation}
  and  there exists a
  Hermitian matrix $\Theta$ such that
 \begin{eqnarray}
\label{JJunitary}
 \tilde F\Theta + \Theta\tilde F^\dagger +
  \tilde GJ\tilde G^\dagger &=&0;\nonumber \\
   \tilde KJ \tilde G^\dagger +
     \tilde H \Theta &=&0.
 \end{eqnarray}
\end{theorem}

\begin{theorem}[See also \cite{ShP5a}.]
\label{T3}
If the QSDEs (\ref{qsde3}), (\ref{tildeFGHK1}) are physically realizable, 
then the corresponding transfer function matrix  (\ref{TF}) is 
 dual $ (J,J)$-unitary. 

Conversely,  suppose the QSDEs (\ref{qsde3}), (\ref{tildeFGHK1}) satisfy the
following conditions: 
\begin{enumerate}[(i)]
\item
The  transfer function matrix  (\ref{TF}) corresponding to
the QSDEs (\ref{qsde3}), (\ref{tildeFGHK1}) is   
 dual $(J,J)$-unitary;
\item
\begin{equation}
\label{KKdagger1}
\tilde K \tilde K^\dagger =I;
\end{equation}
\item
 The Hermitian matrix $\Theta$ satisfying (\ref{JJunitary}) is of the form in
(\ref{Theta}).
\end{enumerate}
Then, the QSDEs (\ref{qsde3}), (\ref{tildeFGHK1}) are physically realizable.
\end{theorem}

\begin{proof}
If the QSDEs (\ref{qsde3}), (\ref{tildeFGHK1}) are physically
realizable, then it follows from Theorem \ref{T1} that 
there exist complex matrices $\Theta =\Theta^\dagger$
and $S$ such that $S^\dagger S = I$ and equations (\ref{physreal1})
are satisfied. However, $\tilde K J \tilde K^\dagger =J$ and $\tilde K \tilde K^\dagger =I$ imply $\tilde K J = J \tilde K$ and hence, it follows from (\ref{physreal1}) that  
\[
\tilde KJ \tilde G^\dagger +
     \tilde H \Theta = 0.
\]
That is, the conditions (\ref{JJunitary}) are satisfied and hence it
follows from Theorem \ref{T2} that the transfer function matrix  (\ref{TF}) corresponding to
the QSDEs (\ref{qsde3}), (\ref{tildeFGHK1}) is   
 dual $ (J,J)$-unitary. 

Conversely, if the QSDEs (\ref{qsde3}), (\ref{tildeFGHK1}) satisfy
conditions (i) - (iii) of the theorem, then it follows from Theorem
\ref{T2} that 
there exists a
  Hermitian matrix $\Theta$ of the form in
(\ref{Theta}) such that equations (\ref{JJunitary}) are
  satisfied. Hence, 
\begin{equation}
\label{io1}
\tilde K J\tilde G^\dagger+\tilde H \Theta = 0.
\end{equation}
Furthermore, we have $\tilde K J \tilde K^\dagger =J$ and $\tilde K \tilde K^\dagger =I$ and therefore  $\tilde K J = J \tilde K$. Hence, (\ref{io1}) implies
\[
\tilde G^\dagger +  \tilde K^\dagger J  \tilde H \Theta=0.
\]
From this it follows that equations (\ref{physreal1}) are
satisfied. Thus, it follows from Theorem \ref{T1} that  the QSDEs (\ref{qsde3}), (\ref{tildeFGHK1}) are physically
realizable.
\end{proof}

\subsection{Physical realizability for annihilator operator  linear quantum systems}
\label{subsec:annihphysrealiz}
For annilhilator operator linear quantum systems described by QSDEs of
the form (\ref{qsde4}) the corresponding formal
definition of physical realizability is as follows.

\begin{definition}(See \cite{MaP1a,MaP3,PET09Aa}.)
\label{D3}
The QSDEs of
the form (\ref{qsde4}) are  said to be {\em physically realizable} if
there exist matrices 
$\Theta_1 = \Theta_1^\dagger >0$,  $\tilde M_1 = \tilde M_1^\dagger$, 
$\tilde N$, and $S$ such that $S^\dagger S = I$ and
(\ref{annihilFGHK}) is satisfied.
\end{definition}

The following theorem from \cite{MaP1a,MaP3,PET09Aa} gives a characterization
of physical realizability in this case.

\begin{theorem}
\label{T4}
The QSDEs (\ref{qsde4}) are physically realizable if
and only if there exists a complex matrix $\Theta_1 = \Theta_1^\dagger
>0$ such that 
\begin{eqnarray}
\label{physreal2}
&&\tilde F \Theta_1 
+ \Theta_1 \tilde F^\dagger + \tilde G \tilde G^\dagger = 0;
\nonumber \\
&&\tilde G =  -\Theta_1 \tilde H^\dagger \tilde K ; \nonumber \\
&&\tilde K^\dagger \tilde K = I.
\end{eqnarray}
\end{theorem}

In the case of QSDEs of the form (\ref{qsde4}), the issue of physical
realizability is determined by the lossless bounded real property of
the corresponding transfer function matrix 
\begin{equation}
\label{annihTF}
\Gamma(s) = \tilde H(sI-\tilde F)^{-1}\tilde G+\tilde K.
\end{equation}
\begin{definition} (See also \cite{AV73}.)
\label{D4}
The transfer function matrix (\ref{annihTF}) corresponding to the
QSDEs  (\ref{qsde4})  is said to be \emph{lossless bounded real} if
the following conditions hold: 
\begin{enumerate}\item[i)] $\tilde F$ is a Hurwitz matrix; i.e., all of its
  eigenvalues have strictly negative real parts; \item[ii)] 
\[
\Gamma(i\omega)^\dagger \Gamma(i\omega)=I
\] for all $ \omega \in
\mathbb{R}.$ 
\end{enumerate}
\end{definition}

\begin{definition}(See also, \cite{MaP1a,MaP3,PET09Aa}.)
\label{D5}
The QSDEs (\ref{qsde4}) are said to define a \emph{minimal} realization of
the transfer function matrix (\ref{annihTF}) if the
following conditions hold:
\begin{enumerate}\item[i)] {\em Controllability}; 
$$\mbox{rank}\left[\begin{array}{ccccc} \tilde G & \tilde F \tilde G& 
\tilde F^2 \tilde G & \ldots & \tilde F^{n-1} \tilde
G \end{array}\right] = n;$$
\item[ii)] {\em Observability};
\[
\mbox{rank}\left[\begin{array}{c} \tilde H \\\tilde H\tilde F\\\tilde
    H\tilde F^2\\ \vdots \\ \tilde H \tilde F^{n-1}\end{array}\right]=n.
\]
\end{enumerate}
\end{definition}

The following theorem, which is a complex version of the standard
lossless bounded real lemma, gives a state space characterization of
the lossless bounded real property.

\begin{theorem}(Complex Lossless Bounded Real Lemma; e.g., see
  \cite{AV73,MaP1a,MaP3}). 
\label{T5}
Suppose the QSDEs (\ref{qsde4}) define  a minimal realization of
the transfer function matrix (\ref{annihTF}). Then the  transfer function
(\ref{annihTF}) is lossless bounded real if and only if there exists a
Hermitian matrix $X>0$ such that 
\begin{eqnarray}
\label{losslessboudedreal}
X\tilde F+\tilde F^\dagger X+\tilde H^\dagger \tilde H &=& 0;\nonumber\\
\tilde H^\dagger \tilde K &=& -X\tilde G;\nonumber \\
\tilde K^\dagger \tilde K &=& I.
\end{eqnarray}
\end{theorem}

Combining Theorems \ref{T4} and \ref{T5} leads to the following result
which provides a complete characterization of the physical
realizability property for minimal QSDEs of the form (\ref{qsde4}). 

\begin{theorem}(See \cite{MaP1a,MaP3,PET09Aa}.)
\label{T6} 
Suppose the QSDEs (\ref{qsde4}) define  a minimal realization of
the transfer function matrix (\ref{annihTF}). 
Then, the  QSDEs (\ref{qsde4}) are physically realizable if and only if
the transfer function matrix (\ref{annihTF}) is lossless bounded real. 
\end{theorem}

The following theorem from \cite{MaP1a,MaP4}, is useful in
synthesizing coherent quantum controllers using state space methods.

\begin{theorem}(See \cite{MaP1a,MaP4}.)
\label{T7}
Suppose the matrices $F, G_1, H_1$ define a minimal realization of the
transfer function matrix
\[
\Gamma_1(s) = H_1(sI-F)^{-1}G_1. 
\]
Then,
there exists matrices $G_2$ and $H_2$ such that the following QSDEs of the form
(\ref{qsde4}) 
\begin{eqnarray}
\label{physreacont}
d\tilde a(t) &=& F \tilde a(t)dt \nonumber \\
&&+ \left[\begin{array}{cc} 
   G_2 & G_1 \\
\end{array} \right]\left[\begin{array}{c}
   d\mathcal{A}_1(t)  \\
   d\mathcal{A}_2(t)
\end{array} \right];\nonumber \\
\left[\begin{array}{c}
    d\mathcal{A}_1^{out}(t)  \\
   d\mathcal{A}_2^{out}(t)
\end{array} \right] &=& \left[\begin{array}{c}
   H_1  \\
   H_2
\end{array} \right]\tilde a(t) dt \nonumber \\
&&+ \left[\begin{array}{cc}
   I & 0\\
   0 & I
\end{array} \right]\left[\begin{array}{c}
   d\mathcal{A}_1(t)  \\
   d\mathcal{A}_2(t)
\end{array}\right]\nonumber \\
\end{eqnarray}
are physically realizable if and only if $F$ is Hurwitz and
\begin{equation} 
\label{physreabound}
\left\| {H_1 \left( {sI - F} \right)^{ - 1} G_1 } \right\|_\infty
\leq 1. 
\end{equation}
\end{theorem}

\subsection{Physical Realizability for Position and momentum operator linear quantum systems}
It is often convenient to consider the physical realizability of real quantum linear systems of the form  (\ref{qsde5}). This can be achieved by applying the transformations (\ref{posmom2}) and the equations (\ref{ABCD1}), (\ref{Jtilde}), (\ref{Jtilde1}), (\ref{DJDT}), (\ref{DDT})   to obtain the following corollary of Theorem \ref{T1}.

\begin{corollary} (See also, \cite{JNP1}, \cite{NJP1}.)
\label{C1}
The QSDEs (\ref{qsde5}) are physically realizable if
and only if there exists a real matrix $\tilde \Theta =-\tilde \Theta^T$
 such that
\begin{eqnarray}
\label{physrealR}
&&A \tilde \Theta 
+ \tilde \Theta  A^T + B\tilde J B^T = 0;
\nonumber \\
&&B =  \tilde \Theta C^T \tilde J D; \nonumber \\
&&D \tilde  J D^T =\tilde J; \nonumber \\
&&DD^T =I.
\end{eqnarray}
\end{corollary}

\begin{remark}
\label{R2}
For real QSDEs of the form (\ref{qsde5}) with corresponding transfer
function 
\[
\Upsilon(s) = C(sI-A)^{-1}B+D
\]
It is straightforward using equations (\ref{posmom2}) to verify that this
transfer function is related to the transfer function (\ref{TF}) of the
corresponding complex QSDEs (\ref{qsde3}) according to the relation
\begin{equation}
\label{TFrelation}
\Upsilon(s) = \Phi \Gamma(s) \Phi^{-1}.
\end{equation}
Now if the real QSDEs (\ref{qsde5}) are physically realizable, it
follows that the corresponding complex QSDEs  (\ref{qsde3}), (\ref{tildeFGHK1}) are physically
realizable. Hence, using Theorem \ref{T3}, it follows  that the corresponding transfer function matrix  (\ref{TF}) is 
 dual $ (J,J)$-unitary; i.e., 
\[
\Gamma(s)J\Gamma^\sim(s) = J
\]
for all $s \in \mathbb{C}_+$. Therefore, it follows from
(\ref{TFrelation}) and (\ref{Jtilde}) that 
\[
\Upsilon(s)\tilde J \Upsilon^\sim(s) = \tilde J
\]
for all $s \in \mathbb{C}_+$. Also, as in the discussion at the end of Section II.C, the conditions (\ref{KJKdagger1}), (\ref{KKdagger1}) are equivalent to the conditions
\begin{eqnarray*}
&&D \tilde  J D^T =\tilde J; \nonumber \\
&&DD^T =I.
\end{eqnarray*}
\end{remark}
\section{Coherent Quantum $H^\infty$ Control}
\label{sec:Hinf}
In this section, we formulate a coherent quantum control problem in
which a linear quantum system is controlled by a feedback controller
which is itself a linear quantum system. The fact that the controller
is to be a quantum system means that any controller synthesis method
needs to produce controllers which are physically realizable.  The
 problem we consider is the quantum $H^\infty$ control problem in
which it is desired to design a coherent controller such that the
resulting closed loop quantum system is stable and attenuates
specified disturbances acting on the system; see
\cite{JNP1,MaP1a,MaP4}. In the standard quantum $H^\infty$ control
problem such as considered in \cite{JNP1,MaP1a,MaP4}, the quantum
noises are averaged out and only the external disturbance is
considered.  
\subsection{Coherent $H^\infty$ control of general quantum linear systems}
\label{subsec:Hinf_gen}
In this subsection, we formulate the coherent quantum $H^\infty$ control
problem for a general class of quantum systems of the form
(\ref{qsde3}), (\ref{tildeFGHK1}). 

We consider quantum \emph{plants} described by linear complex quantum stochastic
models of the
following form defined in an analogous way to the QSDEs
(\ref{qsde3}), (\ref{tildeFGHK1}):

\begin{eqnarray}
\label{plant1}
\left[\begin{array}{l} d\tilde a(t)\\d\tilde a(t)^\#\end{array}\right]
 &=& F \left[\begin{array}{l} \tilde a(t)\\\tilde
     a(t)^\#\end{array}\right]dt\nonumber \\
&&+\left[
{\begin{array}{*{20}c}
   {G_{0 } } & {G_{1 } } & {G_{2}}  \\
\end{array}} \right]\left[ {\begin{array}{*{20}c}
   {dv\left( t \right)}  \\  {dw \left( t \right)} \\ {du\left( t \right)}
\end{array}} \right];  \nonumber\\
dz \left( t \right) &=& H_{1 }\left[\begin{array}{l} \tilde a(t)\\\tilde
     a(t)^\#\end{array}\right] dt 
+ K_{{12} } du \left( t \right); \nonumber\\
 dy \left( t \right) &=& H_{2 } 
\left[\begin{array}{l} \tilde a(t)\\\tilde
     a(t)^\#\end{array}\right] dt \nonumber \\
&&+ \left[ {\begin{array}{*{20}c}
   {K _{{20} } } & {K _{{21} } } & {0 }  \\
\end{array}} \right] \left[ {\begin{array}{*{20}c}
  {dv \left( t \right)}  \\ {dw \left( t \right)}\\ {du \left( t \right)}
\end{array}} \right]\nonumber \\
\end{eqnarray}
where all of the matrices in these QSDEs have a form as in
(\ref{tildeFGHK1}). 
Here, the input 
\[
dw(t)=\left[\begin{array}{l} \beta_w (t)dt+d\mathcal{A}(t)
\\ \beta_w^\# (t)dt+ d\mathcal{A}(t)^{\#} \end{array}\right]
\]
represents a disturbance signal where $\beta_{w}(t)$
is an adapted process; see \cite{JNP1,MaP1a,PAR92}. The signal
$u(t)$ is a control input of the form
\[
du(t)=\left[\begin{array}{l} \beta_u (t)dt+d\mathcal{B}(t)
\\ \beta_u^\# (t)dt+ d\mathcal{B}(t)^{\#} \end{array}\right]
\] 
where  $\beta_{u}(t)$ is an adapted process. The
quantity 
\[
dv(t) = \left[\begin{array}{l} d\mathcal{C}(t)
\\ d\mathcal{C}(t)^{\#} \end{array}\right]
\] 
represents
any additional quantum noise in the plant. The quantities 
 $\left[\begin{array}{l} d\mathcal{A}(t)
\\ d\mathcal{A}(t)^{\#} \end{array}\right]$,
$\left[\begin{array}{l} d\mathcal{B}(t)
\\ d\mathcal{B}(t)^{\#} \end{array}\right]$ and $\left[\begin{array}{l} d\mathcal{C}(t)
\\ d\mathcal{C}(t)^{\#} \end{array}\right]$ are quantum noises of the
form described in Section \ref{sec:systems}.

In the coherent quantum $H^\infty$ control problem, we consider
controllers which are described by QSDEs of the form 
(\ref{qsde3}), (\ref{tildeFGHK1}) as follows:
\begin{eqnarray}
\label{controller1}
\left[\begin{array}{l} d\hat a(t)\\d\hat a(t)^\#\end{array}\right]
 &=& F_c \left[\begin{array}{l} \hat a(t)\\\hat
     a(t)^\#\end{array}\right]dt\nonumber \\
&&+ \left[\begin{array}{*{20}c}
{G_{c_0}} & G _{c_1 } & {G_{c } }
\end{array} \right] 
\left[\begin{array}{*{20}c}
   {dw_{c_0 } }  \\
   {dw_{c_1 } }  \\
   {dy}
\end{array}\right]
  \nonumber\\
\left[\begin{array}{c}
du(t) \\
du_0(t) \\
du_1(t)
\end{array}\right]
 &=& 
\left[\begin{array}{c}
H_c \\
H_{c0} \\
H_{c1}
\end{array}\right]
\left[\begin{array}{l} \hat a(t)\\\hat
     a(t)^\#\end{array}\right] dt  \nonumber \\
&&+\left[\begin{array}{ccc}
K_c & 0 & 0 \\
0 & K_{c0} & 0 \\
0 & 0 & K_{c1}
\end{array}\right]
\left[\begin{array}{*{20}c}
   {dw_{c_0 } }  \\
   {dw_{c_1 } }  \\
   {dy}
\end{array}\right]\nonumber \\
\end{eqnarray}
where all of the matrices in these QSDEs have a form as in
(\ref{tildeFGHK1}). Here the quantities 
 \[
dw_{c_0 } = \left[\begin{array}{l} d\mathcal{A}_c(t)
\\ d\mathcal{A}_c(t)^{\#} \end{array}\right],~~
dw_{c_1 }=\left[\begin{array}{l} d\mathcal{B}_c(t)
\\ d\mathcal{B}_c(t)^{\#} \end{array}\right]
\]
  are controller quantum
noises of the form described in Section \ref{sec:systems}. Also, the
ouputs $du_0$ and $du_1$ are unused outputs of the controller which
have been included so that the controller can satisfy the definition
of physical realizability given in Definition \ref{D1}.

Corresponding to the plant (\ref{plant1}) and (\ref{controller1}), we
form the closed loop quantum system by identifying the output of the
plant $dy$ with the input to the controller $dy$, and identifying the
output of the controller $du$ with the input to the plant $du$. This
leads to the following closed-loop  QSDEs:
\begin{eqnarray} 
\label{closed1}
  d\eta  \left( t \right) &=& \left[ {\begin{array}{*{20}c}
    {F } & {G _{2 } H _{c } }  \\
    {G _{c } H _{2 } } & {F _{c } }  \\
 \end{array}} \right]\eta \left( t \right)dt \nonumber \\
&&+\left[ {\begin{array}{*{20}c}
    {G _{0 } } & {G _{2 }} & 0  \\
    {G _{c } K _{{20} } } & {G_{c_0}} & {G_{c_1}}\\
 \end{array}} \right]\left[ {\begin{array}{*{20}c}
    {dv \left( t \right)}  \\
    {dw_{c_0 } \left( t \right)}  \\
    {dw_{c_1}\left( t \right)}  \\
 \end{array}} \right] \nonumber \\
&&+ \left[ {\begin{array}{*{20}c}
    {G _{1 } }  \\
    {G _{c } K _{{21} } }  \\
 \end{array}} \right]dw \left( t \right);\nonumber \\
   dz \left( t \right) &=& \left[ {\begin{array}{*{20}c}
    {H _{1 } } & {K _{{12} } H _{c } }  \\
 \end{array}} \right]\eta  \left( t \right)dt \nonumber \\
&&+ \left[ {\begin{array}{*{20}c}
    0 & {K _{{12} } }& 0  \\
 \end{array}} \right]\left[ {\begin{array}{*{20}c}
   {dv \left( t \right)}  \\
    {dw_{c_0 } \left( t \right)}  \\
    {dw_{c_1}\left( t \right)}  \\
\end{array}} \right]
\end{eqnarray}
where
\[
\eta \left( t \right) = \left[\begin{array}{l}\tilde a(t)\\
\tilde a(t)^\#\\ \hat a(t)\\\hat
     a(t)^\#\end{array}\right].
\]

For a given quantum plant of the form (\ref{plant1}), the coherent quantum
$H^\infty$ control problem involves finding a physically realizable
quantum controller (\ref{controller1}) such that the resulting closed
loop system (\ref{closed1}) is such that the following conditions are
satisfied:
\begin{enumerate}[(i)]
\item
The matrix
\begin{equation}
\label{cl_hurwitz}
F_{cl} = \left[ {\begin{array}{*{20}c}
    {F } & {G _{2 } H _{c } }  \\
    {G _{c } H _{2 } } & {F _{c } }  \\
 \end{array}} \right]
\end{equation}
is Hurwitz;
\item
 The closed loop transfer function 
\[
\Gamma_{cl}(s) = H_{cl}\left(sI-F_{cl}\right)^{-1}G_{cl}
\]
satisfies
\begin{equation}
\label{cl_hinf}
\|\Gamma_{cl}(s)\|_\infty < 1
\end{equation}
where
\[
H_{cl} = \left[ {\begin{array}{*{20}c}
    {H _{1 } } & {K _{{12} } H _{c } }  \\
 \end{array}} \right], 
~~ G_{cl} = \left[ {\begin{array}{*{20}c}
    {G _{1 } }  \\
    {G _{c } K _{{21} } }  \\
 \end{array}} \right].
\]
\end{enumerate}

\begin{remark}
\label{R3}
In the paper \cite{JNP1}, a version of the coherent quantum $H^\infty$
control problem is solved for linear quantum systems described by real
QSDEs which are similar to those in (\ref{qsde5}). In this case, the
problem is solved using a standard two Riccati equation approach such
as given in \cite{DGKF89,ZDG96}. A result is given in \cite{JNP1}
which shows that any $H^\infty$ controller which is synthesized using the two
Riccati equation approach can be made physically realizable by adding
suitable additional quantum noises. 
\end{remark}

\subsection{Coherent $H^\infty$ control of annihilator operator quantum linear systems}
\label{subsec:Hinf_annih}

In this subsection, we consider the special case of
coherent quantum $H^\infty$ control  for  annihilation operator
quantum linear systems of the 
form considered in Subsection \ref{subsec:annihilator} and present the
Riccati equation solution to this problem obtained in
\cite{MaP1a,MaP3}. The quantum $H^\infty$ control problem being
considered is the same as considered in Subsection
\ref{subsec:Hinf_gen} but we restrict attention to annihilation
operator plants of the form (\ref{qsde4}) as follows:
\begin{eqnarray}
\label{plant2}
d\tilde a(t)
 &=& F \tilde a(t)dt
+\left[
{\begin{array}{*{20}c}
   {G_{0 } } & {G_{1 } } & {G_{2}}  \\
\end{array}} \right]\left[ {\begin{array}{*{20}c}
   {dv\left( t \right)}  \\  {dw \left( t \right)} \\ {du\left( t \right)}
\end{array}} \right];  \nonumber\\
dz \left( t \right) &=& H_{1 }\tilde a(t)dt 
+ K_{{12} } du \left( t \right); \nonumber\\
 dy \left( t \right) &=& H_{2 } 
\tilde a(t) dt 
+ \left[ {\begin{array}{*{20}c}
   {K _{{20} } } & {K _{{21} } } & {0 }  \\
\end{array}} \right] \left[ {\begin{array}{*{20}c}
  {dv \left( t \right)}  \\ {dw \left( t \right)}\\ {du \left( t \right)}
\end{array}} \right].\nonumber \\
\end{eqnarray}
Also, we restrict attention to annihilation
operator controllers of the form (\ref{qsde4}) as follows:
\begin{eqnarray}
\label{controller2}
d\hat a(t)
 &=& F_c \hat a(t)dt\nonumber \\
&&+ \left[\begin{array}{*{20}c}
{G_{c_0}} & G _{c_1 } & {G_{c } }
\end{array} \right] 
\left[\begin{array}{*{20}c}
   {dw_{c_0 } }  \\
   {dw_{c_1 } }  \\
   {dy}
\end{array}\right];
  \nonumber\\
\left[\begin{array}{c}
du(t) \\
du_0(t) \\
du_1(t)
\end{array}\right]
 &=& 
\left[\begin{array}{c}
H_c \\
H_{c0} \\
H_{c1}
\end{array}\right]
\hat a(t) dt  \nonumber \\
&&+\left[\begin{array}{ccc}
K_c & 0 & 0 \\
0 & K_{c0} & 0 \\
0 & 0 & K_{c1}
\end{array}\right]
\left[\begin{array}{*{20}c}
   {dw_{c_0 } }  \\
   {dw_{c_1 } }  \\
   {dy}
\end{array}\right].\nonumber \\
\end{eqnarray}

The quantum plant (\ref{plant2}) is assumed
to  satisfy the following assumptions:
\begin{enumerate}
\item[i)] $K_{{12}}^\dagger K_{{12}}=E_{1}>0;$
\item[ii)] $K_{{21}}K_{{21}}^\dagger=E_{2}>0;$
\item[iii)] The matrix $
\left[ {\begin{array}{*{20}c}
   {F  - i\omega I_n } & {G _{2 } }  \\
   {H _{1 } } & {K _{{12} } }  \\
\end{array}} \right]$ is full rank for all $\omega \geq 0;$
\item[iv)] The matrix $
\left[ {\begin{array}{*{20}c}
   {F  - i\omega I_n } & {G _{1 } }  \\
   {H _{2 } } & {K _{{21} } }  \\
\end{array}} \right]$is full rank for all $\omega \geq 0$.
\end{enumerate}

The results will be stated in terms of the following pair of complex algebraic Riccati equations:
\begin{eqnarray}\label{riccati1}
&& \left( {F  - G _{2 } E_{1 }^{ - 1} K _{{12} }^\dagger H _{1 } }
\right)^\dagger X  + X \left( {F  - G _{2 } E_{1 }^{ - 1} K _{{12}
    }^\dagger H _{1 } } \right) \nonumber \\
&&+ X \left( {G _{1 } G _{1 }^\dagger  -  G _{2 } E_{1 }^{ - 1} G _{2
    }^\dagger } \right)X  \nonumber \\ 
&&+    H _{1 }^\dagger \left( {I - K
    _{{12} } E_{1 }^{ - 1} K _{{12} }^\dagger } \right)H _{1 }  = 0;
 \\ \label{riccati2}
&& \left( {F  - G _{1 } K _{{21} }^\dagger E_{2 }^{ - 1} H _{2 } }
\right) Y  + Y \left( {F  - G _{1 } K _{{21} }^\dagger E_{2 }^{ - 1} H
    _{2 } } \right)^\dagger \nonumber \\
&&+ Y \left( {H _{1 }^\dagger H _{1 }-  H
    _{2 }^\dagger E_{2 }^{ - 1} H _{2 }} \right)Y \nonumber \\ && + G _{1 } \left( {I - K _{{21}}^\dagger  E_{2 }^{ - 1} K _{{21} } } \right)G _{1 }^\dagger  = 0.
\end{eqnarray}
The solutions to these Riccati equations will be required to satisfy
the following conditions.
\begin{enumerate}
  \item[i)] The matrix $F  - G _{2 } E_{1 }^{ - 1} K _{{12} }^\dagger H _{1 }+ \left( {G _{1 } G _{1 }^\dagger  - G _{2 } E_{1 }^{ - 1} G _{2 }^\dagger } \right)X $ is Hurwitz; i.e., $X$ is a stabilizing solution to (\ref{riccati1}).
  \item[ii)] The matrix $F  - G _{1 } K _{{21} }^\dagger E_{2 }^{ - 1} H _{2 }+Y \left( { H _{1 }^\dagger H _{1 }-  H _{2 }^\dagger E_{2 }^{ - 1} H _{2 }} \right)$ is Hurwitz; i.e, $Y$ is a stabilizing solution to (\ref{riccati2}).
  \item[iii)]
The matrices $X$ and $Y$ satisfy 
\begin{equation}
\label{spec_radius}
\rho(XY) < 1
\end{equation}
where $\rho(\cdot)$ denotes the spectral radius. 
\end{enumerate}

If the above Riccati equations have suitable solutions, a quantum
controller of the form (\ref{controller2}) is constructed as follows:
\begin{eqnarray}
\label{central}\nonumber
 F_{c }  &=& F  + G _{2 } H _{c }  - G _{c } H _{2 }  + \left( {G _{1 }  - G _{c } K _{{21} } } \right)G _{1 }^\dagger X; \\\nonumber
 G_{c }  & =& \left( {I - Y X } \right)^{ - 1} \left( {Y H _{2 }^\dagger  + G _{1 } K _{{21} }^\dagger } \right)E_{2 }^{ - 1};\\
H _{c }  &=&  - E_{1 }^{ - 1} \left( {g^2 G _{2 }^\dagger X  + K _{{12} }^\dagger H _{1 } } \right).
\end{eqnarray}

The following Theorem is presented in \cite{MaP1a,MaP3}.

\begin{theorem}
\label{T8} 
\emph{Necessity:} Consider a quantum plant (\ref{plant2}) satisfying
the above assumptions. If there exists a
quantum  controller of the form (\ref{controller2}) such that the resulting closed-loop
  system satisfies the conditions (\ref{cl_hurwitz}), (\ref{cl_hinf}), then
  the Riccati equations (\ref{riccati1}) and (\ref{riccati2}) will have stabilizing solutions
  $X \geq 0$ and $Y \geq 0$ satisfying (\ref{spec_radius}). \\ 
\emph{Sufficiency:} Suppose the Riccati equations (\ref{riccati1}) and (\ref{riccati2}) have stabilizing solutions $X\geq 0$ and $Y \geq 0$ satisfying (\ref{spec_radius}). If the controller (\ref{controller2}) is such that the matrices
$F_{c}$, $G_{c}$, $H_{c}$ are as defined in (\ref{central}), then the
resulting closed-loop system  will satisfy the conditions
(\ref{cl_hurwitz}), (\ref{cl_hinf}).
\end{theorem}

Note that this theorem does not guarantee that a controller defined
by (\ref{controller2}), (\ref{central}) will be physically realizable. However, if the
matrices defined in (\ref{central}) are such that 
\[
\|H_c\left(sI-F_c\right)^{-1}G_c\| < 1,
\]
then it follows from Theorem \ref{T7} that a corresponding physically
realizable controller of the form (\ref{controller2}) can be
constructed. 
\section{Conclusions}
\label{sec:conclusions}
In this paper, we have surveyed some recent results in the area of
quantum linear systems theory and the related area of coherent quantum
$H^\infty$ control. However, a number of other recent results on
aspects of quantum linear systems theory have not been covered in
this paper. These include results on coherent quantum LQG control (see
\cite{JNP1,HP4a}), and model reduction for quantum linear systems (see
\cite{PET09Aa}). Furthermore,  in order to apply 
synthesis results on coherent quantum feedback controller synthesis, it is
necessary to realize a synthesized feedback controller transfer
function using physical
optical components such as optical cavities, beam-splitters, optical
amplifiers, and phase shifters. In a recent paper \cite{NJD09}, this
issue was addressed for a general class of coherent linear quantum
controllers. An alternative approach to this problem is addressed in
\cite{PET08Aa} for the class of annihilation operator linear
quantum systems considered in Subsection \ref{subsec:annihilator} and
\cite{MaP1a,MaP3,MaP4}. For this class of quantum systems, an algorithm is given
to realize a physically realizable controller transfer function in
terms of a cascade connection of optical cavities and phase
shifters. 

An important application of   both classical
and coherent  feedback control of quantum systems is 
in enhancing the property of entanglement for linear quantum
systems. Entanglement is an intrinsically quantum mechanical notion
which has many applications in the area of quantum computing and
quantum communications.

To conclude, we have surveyed some of the important advances
in the area of linear quantum control theory. However,  many important
problems  in this area remain
open and the area provides a great scope for future research.


\end{document}